\documentclass[11pt]{article}
\usepackage{amsmath}
\usepackage{amsfonts}
\usepackage{amssymb}
\usepackage{amsthm}
\usepackage{bbm}
\usepackage{color}
\usepackage{mathrsfs}
\usepackage{complexity}
\usepackage{float}
\usepackage{algorithm}
\floatname{algorithm}{Protocol}
\usepackage{multirow}
\usepackage{graphicx}
\usepackage[margin=0.8in]{geometry}
\usepackage[font=small,margin=0.5in]{caption}
\usepackage[normalem]{ulem}
\usepackage{mathtools}
\usepackage{qcircuit}

\newcommand{\COMMENT}[1]{}
\newcounter{linenum}

\definecolor{bluegreen}{cmyk}{0.85,0.3,0.2,0.25}

\newcommand{\ket}[1]{\left\lvert #1 \right\rangle}

\usepackage{color}

\newtheorem{theorem}{Theorem}

\newtheorem{corollary}{Corollary}

\newtheorem{definition}{Definition}

\usepackage[affil-it]{authblk}

\begin{document}
\graphicspath{{diagrams/}}

\title{Multiparty Delegated Quantum Computing}
\author[1,2]{Elham Kashefi}
\author[1,3]{Anna Pappa}
\affil[1]{School of Informatics, University of Edinburgh, UK}
  \affil[2]{LIP6, CNRS, University Pierre et Marie Curie, France}
  \affil[3]{Department of Physics and Astronomy, University College London, UK}

\maketitle
\abstract{\noindent Quantum computing has seen tremendous progress in the past years. However, due to limitations in scalability of quantum technologies, it seems that we are far from constructing universal quantum computers for everyday users. A more feasible solution is the delegation of computation to powerful quantum servers on the network. This solution was proposed in previous studies of Blind Quantum Computation, with guarantees for both the secrecy of the input and of the computation being performed. In this work, we further develop this idea of computing over encrypted data, to propose a multiparty delegated quantum computing protocol in the measurement-based quantum computing framework. We prove security of the protocol against a dishonest Server and against dishonest clients, under the asumption of common classical cryptographic constructions.}

\vspace{0.5in}

\section{Introduction}

Since the early days of quantum computing and cryptography, research has been focused on finding secure communication protocols for different cryptographic tasks. However, the no-go results for bit commitment \cite{LoChau,Mayers} and oblivious transfer \cite{Lo97} soon provided evidence that it is not possible to guarantee perfect security against any type of quantum adversaries. In fact, the authors of \cite{SSS09} showed that any non-trivial protocol that implements a cryptographic primitive, necessarily leaks information to a dishonest player. It directly follows that two-party unitaries, and by consequence multi-party ones, cannot be used without further assumptions, to securely implement cryptographic protocols. An important question that rises then is what are the cryptographic assumptions that are needed in order to achieve secure multiparty computation.

Dupuis et al \cite{DNS10,DNS12} examined the case of two-party computation and showed that access to an AND gate for every  gate that is not in the Clifford group, plus a final SWAP gate are required, in order to guarantee security. In the multiparty setting, Ben-Or et al \cite{Benor} rely on an honest majority assumption in order to build a Verifiable Quantum Secret Sharing scheme that is the basis of the multiparty quantum computation. From a different perspective, a lot of research in quantum computing has been focused on delegation of computation to powerful servers. This is because the current state-of-the-art is still far from constructing scalable quantum devices, and it seems that the first quantum networks will rely on the use of a limited number of powerful quantum servers. Common requirements from delegated computation schemes are that they are provably secure (the input of the computation remains private), blind (the computation performed is hidden from the Server) and verifiable (honest participants can verify the correctness of the computation).

In this work, we extend previous research on quantum computing, by examining the problem of secure delegation of a multiparty quantum computation to a powerful Server. More specifically, we suppose that a number of clients holding some quantum input, want to perform a unitary operation on the state, but are lacking the computational abilities to do so, therefore would like to delegate the computation to a Server. In the proposed protocol, the quantum operations required from the clients are limited to  creating $\ket{+}$ states and applying X gates and rotations around the $\hat{z}$-axis. To be secure against coalitions of dishonest clients, all participants should contribute to some form of quantum encryption process.  In previous protocols~\cite{Benor}, quantum communication between all clients was required in order to provide security against dishonest participants. However, by using a remote state preparation procedure, we manage to remove any quantum communication between clients, making our protocol adaptable to a client/server setting. More interestingly, the quantum communication from the clients to the Server can be done in single-qubit rounds, not necessitating any quantum memory from the clients. Furthermore, all quantum communication takes place in the preparation (offline) phase, which makes the computation phase more efficient, since only classical communication is required. 

As already mentioned, in order to provide any type of security in the multiparty setting, we need to make some assumptions about the dishonest parties. In this work, we will need two assumptions. First, we will assume that the clients have secure access to classical multiparty functionalities, which we will treat as oracles. This is a common construction in classical secure multiparty computation, and uses assumptions on the participating parties, like honest majority or difficulty to invert specific one-way functions. The second assumption is that a set of malicious clients cannot corrupt the Server, and the other way around. This means that we only prove security against two adversarial models, against a dishonest Server, and against a coalition of dishonest clients. Security in the more general scenario where a Server and some clients collaborate to cheat, remains as an open question (however see \cite{Wallden16} for a relevant model with only one client, where the Server is also allowed to provide an input).

Finally, we should note that in this work, we are focusing on proving security against malicious quantum adversaries in order to provide a simple protocol for quantum multiparty computation. As such, no guarantee is given on the correctness of the computation outcome. However, in principle, it might be possible to add verification processes in our protocol, by enforcing honest behaviour, following the work of \cite{DNS12} and \cite{FK13}.

\section{Results}

We propose a cryptographic protocol that constructs a Multiparty Delegated Quantum Computing resource using quantum and classical communication between $n$ clients and a Server. We want to guarantee that the private data of the clients remain secret during the protocol. Here, each client's  data consists of his quantum input and output, while we consider that the computation performed is not known to the server (it can be known to all clients, or to a specific client that delegates the computation). The protocol consists of two stages, a preparation one, where all the quantum communication takes place, and the computation one, where the communication is purely classical. During the preparation stage, a process named ``Remote State Preparation'' \cite{DKL12} (see Figures \ref{fig:algo1} and \ref{fig:algo2}) asks the clients to send quantum states to the Server, who then entangles them and measures all but one. This process allows the clients to remotely prepare quantum states at the Server's register that are ``encrypted'' using secret data from all of them, without having to communicate quantum states to each other. 

The above process however could allow the clients to affect the input of the other clients by changing the classical values they use throughout the protocol. We therefore ask all clients to commit to using the same classical values for the duration of the protocol by verifiably secret-sharing them. Further, to check that the clients are sending the correct quantum states at the preparation phase of the protocol, they are asked to send many copies of them, which are then (all but one) checked for validity. 

At the end of the remote state preparation phase, the Server is left with several quantum states that are ``encrypted'' using secret data from all the clients. These states will be used to compute the desired functionality, in the Measurement-based Quantum Computing (MBQC) framework. Due to the inherent randomness of the quantum measurements that are happening throughout the computation, there is an unavoidable dependency between the measurement angles of the qubits in the different layers of computation. This means that the clients need to securely communicate between them and with the Server, in order to jointly compute the updated measurement angles, taking into account the necessary corrections from the previous measurements according to the dependency sets of each qubit. This procedure is purely classical, and uses Verifiable Secret Sharing (VSS) schemes and a computation oracle to calculate the necessary values at each step of the protocol and to ensure that the clients behave honestly.

Finally, in the output phase, each output qubit is naturally encrypted due to the same randomness from previous measurements that propagated during the computation. Each client receives the  appropriate encrypted output qubit, and computes the values necessary for the decryption from the previously shared values of all clients. 

In comparison with previous work on two-party \cite{DNS12} and multiparty \cite{Benor} computation, our protocol does not require that the clients possess quantum memories, measurement devices or entangling gates, but that they only perform single qubit operations.  This practical aspect of our protocol will be useful in near-future hybrid quantum-classical networks, since clients with limited quantum abilities will be able to delegate heavy computations to a powerful quantum Server. 
The security of our protocol is proven in two settings, against a dishonest Server, and a coalition of malicious clients. It is also straightforward to compose it with other protocols, either in parallel or sequentially, since security is defined in a composable framework. It remains to study whether the proposed protocol remains secure against a dishonest coalition between clients and the Server or if there is an unavoidable leakage of information.

\section{Materials and Methods}

\subsection{Measurement-based Quantum Computing}

Delegated computation is commonly studied in the MBQC model~\cite{RB01}, where a computation is described by a set of measurement angles on an entangled state. A formal way to describe how MBQC works was proposed in \cite{DK06}, and is usually referred to as an MBQC pattern. In the general case of quantum input and quantum output, such a pattern is defined by a set of qubits ($V$), a subset of input qubits ($I$), a subset of output qubits ($O$), and a sequence of measurements $\{\phi_j\}$ acting on qubits in $O^c:=V\setminus O$. Due to the probabilistic nature of the measurements, these angles need to be updated according to a specific structure. This structure is described by the \emph{flow} $f$ of the underlying graph $G$ of the entangled state. The flow is a function from measured qubits to non-input qubits along with a partial order over the nodes of the graph such that each qubit $j$ is X-dependent on qubit $f^{-1}(j)$ and Z-dependent on qubits $i$ for which $j\in N_G(f(i))$, where $N_G(j)$ is the set of neighbours of node $j$ in graph $G$. We will denote the former set of qubits by $S^X_j$ and the latter set of qubits by $S^Z_j$. The computation of a unitary $U$ on quantum input $\rho_{in}$, described by angles $\{\phi_j\}$ is shown in Figure~\ref{fig:z-prog circuit}.

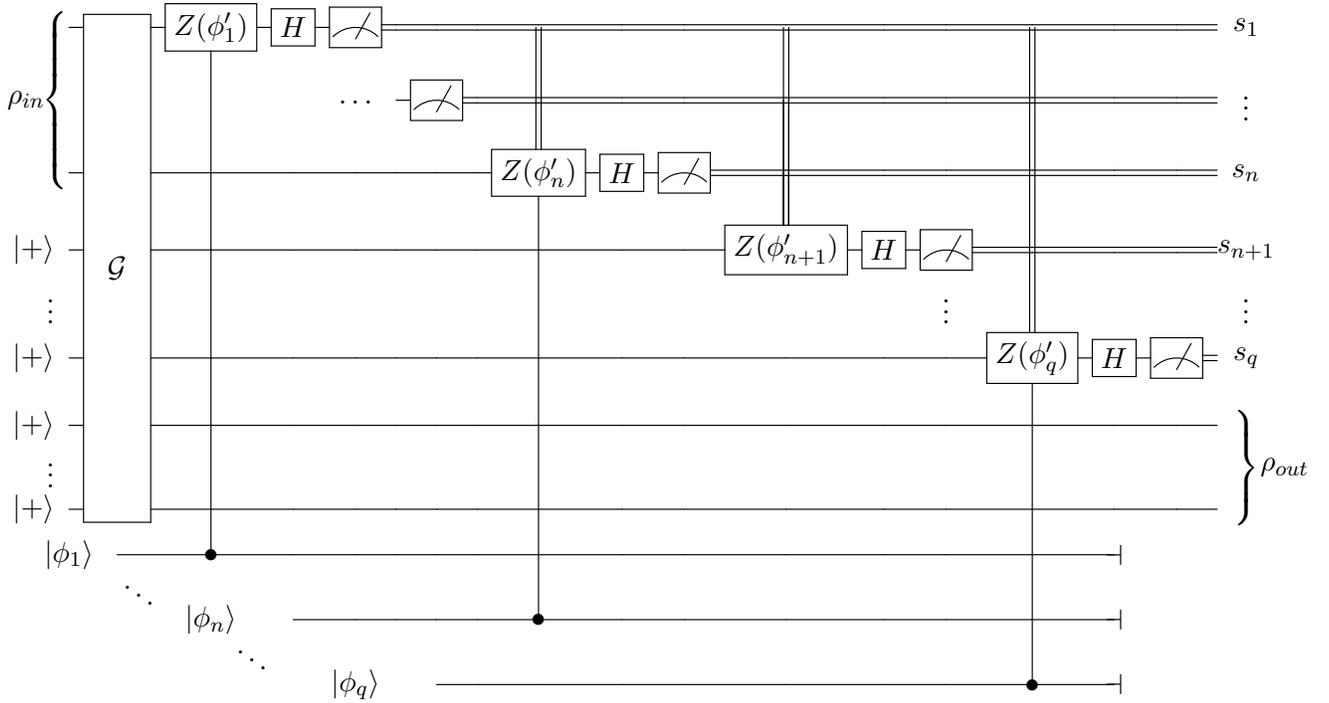
\begin{figure}[!h]

\centerline{
\Qcircuit @C=0.5em @R=1em {
				&\multigate{8}{\ \mathcal{G}\ } 	&\gate{Z(\phi'_1)}&\gate{H} &\meter	&	\cw 	&\cw		&\cw	&\cw \cwx[2]	&\cw &\cw &\cw \cwx[3] &\cw &\cw &\cw \cwx[5] & \cw & \cw & \cw & & s_1\\	
\lstick{\rho_{in} \ }		&	&				 &		   &\dots		&		 &	\meter 	&	\cw		&	\cw	&\cw &\cw &\cw \cwx[2] &\cw &\cw	&\cw &\cw &\cw	&\cw & & \vdots\\
			 	&\ghost{\ \mathcal{G}\ }	&\qw	&\qw			 &\qw	   &\qw		&\qw &\qw & \gate{Z(\phi'_n)}	&\gate{H}	&\meter & \cw \cwx[1]&\cw  &\cw & \cw & \cw & \cw & \cw & & s_n\\
\lstick{\ket{+}}		&		\ghost{\ \mathcal{G}\ }	&\qw   			 &\qw 	   &\qw		& \qw	 &\qw		&\qw	&\qw	&\qw	&\qw	&\gate{Z(\phi'_{n+1})}&\gate{H}	&\meter 	&\cw	&\cw & \cw & \cw & & s_{n+1}\\
\lstick{\vdots}  		&	&				 &		   &		&		 &			&		&		&	& & & &\vdots		& & &  & & & \vdots\\
\lstick{\ket{+}} 		&		\ghost{\ \mathcal{G}\ }	&\qw   			 &\qw 	   &\qw		& \qw	 &\qw		&\qw		&\qw	&\qw	& \qw	&\qw	&\qw & \qw & \gate{Z(\phi'_q)}&\gate{H}	&\meter 	&\cw	& 	& s_q\\
\lstick{\ket{+}} 		&		\ghost{\  \mathcal{G}\ }	&\qw   			 &\qw 	   &\qw		& \qw	 &\qw		&\qw		&\qw	&\qw	& \qw	&\qw	&\qw & \qw & \qw	& \qw	&\qw	& \qw	& &\\
\lstick{\vdots}  		&	&				 &		   &		&		 &			&		&		&	& & & &		& & &  & & & & &~~\rho_{out}\\
\lstick{\ket{+}} 		&		\ghost{\ \mathcal{G}\ }	&\qw   			 &\qw 	   &\qw		& \qw	 &\qw		&\qw		&\qw	&\qw	& \qw	&\qw	&\qw & \qw & \qw	&\qw 	&\qw	& \qw	& &\\
\ket{\phi_1}	 		& 	&\ctrl{-9} 	 &	\qw	   &	\qw	 &\qw	 &\qw 	&	\qw		&\qw	&\qw & \qw		&\qw	&\qw	&	\qw		&\qw	&\qw \dashv \\
  				& ~~~~ \ddots			&	 &		   &		&		 &			&			&		&		&\\
				 & 	&\ket{\phi_n}			 & & \qw & \qw & \qw &\qw		 &\ctrl{-9}		 &	\qw &	\qw		&\qw	&\qw	 	&		\qw	&	\qw	&	\qw	\dashv & \\
  				&	&	~~~~~~~~ \ddots		 &		   &		&		 &			&			&		&		& & &\\				 
				 &	&						 &  &\ket{\phi_q}& 		  &  & \qw &\qw &	\qw &\qw	 &	\qw	 &	\qw	 &	\qw		&	\ctrl{-8}		&	\qw \dashv
\gategroup{1}{1}{3}{1}{1em}{\{}				
\gategroup{7}{19}{9}{19}{1em}{\}}				
				 }		
	}
\caption{A circuit description of the Measurement-based Quantum Computation model}\label{fig:z-prog circuit}
\end{figure}

\subsection{Abstract Cryptography}
We will use the Abstract Cryptography framework \cite{MauRen11} to model the properties of our protocol. This is a top-down approach that defines in an abstact way the components of a system and its interactions with the environment. In a similar way as most security frameworks, the goal in AC is to prove that a functionality in the real world, is indistinguishable from a functionality in the ideal world. In the real world, the parties interact through communication channels, while in the ideal world the parties have access to the ideal system that computes the functionality.

We define a \emph{resource} (denoted by capital letters $\mathcal{W},\mathcal{V}$ etc) as a system with a set of interfaces $\mathcal{I}=\{1,\dots,n\}$. This set of interfaces corresponds to the parties involved in a protocol and provides them with functionalities like inputting data and receiving outputs. In general, we can interpret any cryptographic protocol as an effort to construct an ideal functionality (i.e. resource) from other components. For this, we need to introduce the notion of a \emph{converter}. Converters are systems with two interfaces, an inside and an outside interface: the outside gets inputs and gives outputs, and the inside is connected to a resource. 
A protocol $\pi$ can be viewed as a set of converters $\{\pi_i\}_i^n$ where $n$ is the number of participants in the protocol. Therefore, a protocol connected to a communication channel resource, can be thought of as a new resource. Another type of converters commonly used in security proofs are called \emph{simulators}.  A simulator $\sigma$ is used in order to reproduce the behaviour of some other resource, therefore making it easier to compare between two different resources. 

Finally, we define a \emph{filtered resource} as a pair of a resource and a converter that blocks malicious activity. For example, the filtered resource $\mathcal{W}_f$ is comprised of the resource $\mathcal{W}$ and the converter $f$ that blocks all malicious activity from dishonest parties, leaving only legitimate input through to the resource $\mathcal{W}$.

\subsubsection{Security}
We can measure how close two resources are by defining a pseudo-metric $d$ on the space of the resources. Any chosen pseudo-metric needs to have three properties: $d(\mathcal{V},\mathcal{V})=0$ (identity), $d(\mathcal{V},\mathcal{W})=d(\mathcal{W},\mathcal{V})$ (symmetry) and $d(\mathcal{V},\mathcal{W})\leq d(\mathcal{V},\mathcal{Y})+d(\mathcal{Y},\mathcal{W})$ (triangle inequality). Here we will use the \emph{distinguishing advantage} of a distinguisher $\mathcal{Z}$ as a measure of distance between two resources: 
\[
d(\mathcal{V},\mathcal{W}):=max_{\mathcal{Z}}\{\Pr[\mathcal{ZV}=1]-\Pr[\mathcal{ZW}=1]\}
\]
The distinguisher $\mathcal{Z}$ is a system that has as many inside interfaces connecting to the resource as the number of parties communicating with the latter, and an extra one on the outside to output a single bit. Simply put, the above measure of distance shows how well can a distinguisher understand whether he is interacting with resource $\mathcal{V}$ or $\mathcal{W}$. The reason why a distinguisher is used to quantify security is \emph{composability}. Previous studies of secure quantum computation proved security in the stand-alone model, which does not directly imply that security is maintained during parallel execution of protocols or in cases where there is a subsequent leakage of information to dishonest participants. A prominent example from Quantum Key Distribution (QKD) is demonstrated in~\cite{PR14,KRBM05}. There, it is shown that the security definitions requiring that the mutual information of the produced key and the adversary's system is small, are not adequate when the adversary's system and the key are correlated, since a partial leakage of the key in subsequent steps can lead to the compromise of the whole key.

A general security definition for constructing one resource from another, is given in \cite{DFPR14}:

\begin{definition}\label{security}
Let $\mathcal{V}$, $\mathcal{W}$ be two resources with interfaces in set $\mathcal{I}$. A protocol $\pi=\{\pi_i\}_{i\in\mathcal{I}}$ securely $\epsilon$-constructs resource $\mathcal{V}$  from $\mathcal{W}$  if there exist converters $f=\{f_i\}_{i\in\mathcal{I}}$, $g=\{g_i\}_{i\in\mathcal{I}}$ and $\sigma=\{\sigma_i\}_{i\in\mathcal{I}}$ such that:
\begin{equation}\label{distance}
\forall \mathcal{P}\subseteq \mathcal{I}, ~~~~~d(\pi_{\mathcal{P}}g_{\mathcal{P}}\mathcal{V},\sigma_{\mathcal{I}\backslash\mathcal{P}}f_{\mathcal{P}}\mathcal{W})\leq \epsilon
\end{equation}
where for each converter $x=\{x_i\}_{i\in\mathcal{I}}$, we define $x_{\mathcal{P}}:=\{x_i\}_{i\in\mathcal{P}}$.
\end{definition}

\subsection{Multiparty Delegated Quantum Computing}
In this section, we will give some necessary definitions for multiparty delegated quantum computing protocols. We will consider multiple clients $C_1,...,C_n$ that have registers $\mathcal{C}_1,..,\mathcal{C}_n$.  To allow the computation to be done in a delegated way, we also introduce the notion of a \emph{Server} $S$, who is responsible for performing the computation of a unitary $U$ on input $\rho_{in}$. The Server has register $\mathcal{S}$, but no legal input to the computation. We also denote with $\mathcal{D}(\mathcal{A})$ the set of all possible quantum states (i.e. positive semi-definite operators with trace 1) in register $\mathcal{A}$. We denote the input state (possibly entangled with the environment $R$) as:

 \begin{equation}\label{eq:input_state}
 \rho_{in}\in \mathcal{D}(\mathcal{C}_1\otimes\dots\otimes\mathcal{C}_n\otimes \mathcal{R})
 \end{equation}

In what follows, we consider that the input registers of the participants also contain the classical information necessary for the computation. We denote with $\mathcal{L}(\mathcal{A})$ the set of linear mappings from $\mathcal{A}$ to itself, and  we call a superoperator $\Phi:\mathcal{L}(\mathcal{A})\rightarrow \mathcal{L}(\mathcal{B})$ that is completely positive and trace preserving, a \emph{quantum operation}. Finally, we will denote with $\mathbb{I}_{\mathcal{A}}$ and $\mathbf{1}_{\mathcal{A}}$ the totally mixed state and the identity operator respectively, in register $\mathcal{A}$.

In order to analyse the properties of the protocol to follow, we will first consider an abstract system that takes as input $\rho_{in}$ and the computation instructions for implementing $U$ in MBQC (i.e. the measurement angles $\{\phi_j\}$), and outputs a state $\rho_{out}$.  We will call such a resource an \emph{Ideal Functionality} because this is what we want to implement. In order to allow the Server to act dishonestly, we also allow the Ideal Functionality to accept input from the Server, that dictates the desired deviation in the form of quantum input and classical information (Figure~\ref{fig:MPQC}). In the case where the Server is acting honestly, the Ideal Functionality is outputting the correct output $\rho_{out}=(U\otimes\mathbf{1}_{\mathcal{R}})\cdot\rho_{in}$, where for ease of use, we will write $U\cdot \rho$ instead of $U\rho U^{\dagger}$ each time we talk about applying a unitary operation $U$ to a quantum state $\rho$. 

\begin{definition}[MPQC Delegated Resource]\label{def:Ideal}
A \emph{Multiparty Quantum Computation delegated resource} gets input $\rho_{in}$ and computation instructions from $n$ clients and leaks the size of the computation $q$ to the Server.  The Server can then decide to input a quantum map and a quantum input. The output of the computation is $\rho_{out}=(U\otimes \mathbf{1}_{\mathcal{R}})\cdot\rho_{in}$ if the Server does not deviate, otherwise the output is defined by the quantum deviation on the inputs of the Server and the clients.
\end{definition}

\begin{figure}
\begin{center}
\includegraphics[scale=0.8]{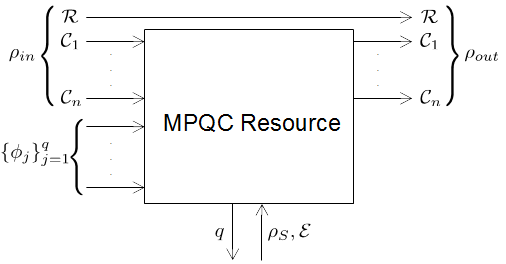}
\caption{The Ideal functionality that implements a Multiparty Quantum Computation given by measurement angles $\{\phi_j\}_{j=1}^q$ on input $\rho_{in}$.}\label{fig:MPQC}
\end{center}
\end{figure}

A protocol can in general be modeled by a sequence of local quantum operations on the participants' registers together with some oracle calls, which are joint quantum operations on the registers. Here we will consider the delegated version of communication protocols \cite{DNS10,DNS12,Benor}, where $n$ clients are delegating the computation to a Server.

\begin{definition}
We denote a $t$-step delegated protocol with oracle calls between clients $C_k$ $(k\in[n])$ and a Server $S$, with $\pi\mathcal{O} =(\{\pi_k\}_{k=1}^n,\pi_S, \mathcal{O})$. If we denote by $\mathcal{C}_k$ and $\mathcal{S}$ the registers of the clients and the Server respectively, then:
\begin{itemize}
\item Each client's strategy $\pi_k,~k\in[n]$ consists of a set of local quantum operators $(L_1^k,\dots,L_t^k)$ such that $L_i^k:\mathcal{L}(\mathcal{C}_k)\rightarrow \mathcal{L}(\mathcal{C}_k)$ for $1\leq i \leq t$.
\item The Server's strategy $\pi_S$ consists of a set of local quantum operators $(L_1^S,\dots,L_t^S)$ such that $L_i^S:\mathcal{L}(\mathcal{S})\rightarrow \mathcal{L}(\mathcal{S})$ for $1\leq i \leq t$.
\item The oracle $\mathcal{O}$ is a set of global quantum operators $\mathcal{O} = (\mathcal{O}_1,\dots,\mathcal{O}_t)$ such that $\mathcal{O}_i : \mathcal{L}(\mathcal{C}_1\otimes\dots\otimes \mathcal{C}_n\otimes\mathcal{S})\rightarrow \mathcal{L}(\mathcal{C}_1\otimes\dots\otimes \mathcal{C}_n\otimes\mathcal{S})$ for $1 \leq i \leq t$.
\end{itemize}
\end{definition}
\vspace{0.2in}
\noindent{}Therefore, at each step $i$ of the protocol $\pi\mathcal{O}$, all participants apply their local operations and they also jointly apply a global $\mathcal{O}_i$. If we define the joint quantum operation $L_i=L_i^1\otimes \dots\otimes L_i^n\otimes L_i^S$, then the quantum state at step $i$ when the input state $\rho_{in}$ is defined by Equation \ref{eq:input_state}, is $\rho_i(\rho_{in}):=(\mathcal{O}_i L_i\dots\mathcal{O}_1 L_1)\cdot \rho_{in}$. At the end of the protocol, the output will be: $\pi\mathcal{O}(\rho_{in})=(\mathcal{O}_t L_t\dots \mathcal{O}_1 L_1)\cdot \rho_{in}$.

A commonly used oracle is the \emph{communication oracle} that transmits information from one participant to the other, by just moving the state between registers. Due to no-cloning however, when the information transmitted is a quantum state, we require that the oracle also erases the transmitted information from the original party's register. 

Another oracle that we will use in our protocol is what we call the \emph{computation oracle} which can be thought of as a classical ideal functionality (i.e. a box that takes inputs from all parties and outputs the correct outputs of the functionality on the speific inputs). All classical multiparty computation in this work will be done with the help of such a \emph{computation oracle}. Under standard cryptographic assumptions, there exist classical protocols for building such oracles that emulate in a composable way any classical multiparty functionality. In~\cite{Canetti01} this is shown for an honest majority of participants, given that they share pairwise secure classical channels and a broadcast channel. A similar result based on Oblivious Transfer appears in~\cite{Ishai08}. The quantum lifting theorem of Unruh~\cite{Unruh10} states that if there exists such a construction that is secure against classical adversaries, then the same also holds against quantum adversaries, therefore allowing us to replace any classical multiparty computation in a quantum protocol, by using the computation oracle. This result was proven in the Universal Composability framework, but we can also use it in Abstract Cryptography proofs, since the two frameworks are equivalent when there is only one malicious participant (i.e. we need to build one global simulator).

\subsubsection{Properties}
In the following Section, we will present a delegated protocol $\pi\mathcal{O}$ that emulates the MPQC resource of Figure~\ref{fig:MPQC}. Definition \ref{security} gives $2^n$ inequalities to be satisfied for such a protocol to be secure against all possible sets of dishonest participants. However we can group them into fewer distinct cases according to our adversarial setting: (a) when all clients and the Server are honest, (b) when the Server is dishonest, (c) when a subset $\mathcal{D}$ of clients is dishonest. For the case (a) where everyone acts honestly, Definition \ref{security} for a filtered MPQC resource $M_\perp$ becomes:

\begin{equation}\label{correct}
d(\pi \mathcal{O}, \mathcal{M}_\perp)\leq \epsilon
\end{equation}
In case (b) when the Server is acting dishonestly, Definition \ref{security} becomes:
\begin{equation}\label{secure1}
d(\pi_1\dots\pi_n\mathcal{O},\sigma_S \mathcal{M})\leq \epsilon
\end{equation}
where $\sigma_S$ is a simulator for the dishonest server. Finally for case (c), when a subset $\mathcal{D}$ of clients is dishonest, Definition \ref{security} becomes:
\begin{equation}\label{secure2}
d(\pi_{\mathcal{H}}\mathcal{O}\pi_S,\sigma_{\mathcal{D}}\mathcal{M}_\perp)\leq \epsilon
\end{equation}
where $\mathcal{H}$ is the set of honest client and $\sigma_{\mathcal{D}}$ is a simulator for the dishonest clients. 

In the above equations, the filter $\perp$ blocks input from the Server's side, since the MPQC resource of Definition~\ref{def:Ideal} accepts deviated behaviour only from the Server. The clients on the other hand, have the liberty to choose the quantum state that they will give as input to the MPQC resource. However, this does not mean that the clients behave honestly during the protocol. We will see that the protocol ``enforces'' honest behaviour to clients by asking them to secret share their classical values in a verifiable way, in order to commit to use the same values during the protocol. This is done using Verifiable Secret Sharing (VSS) schemes, that allow a dealer to share their information with other parties, in such a way that a group of honest parties can reconstruct the share even when the dealer is malicious. VSS schemes can be viewed as a multiparty computation, and we can therefore use a computation oracle under the same cryptographic assumptions as discussed above. 

The distinghishing advantage of Equations~(\ref{correct}), (\ref{secure1}) and (\ref{secure2}) can reduce to simple measures of distance between the states that a distinguisher sees, when interacting with the real and ideal system. For example if the outputs of the resources are classical strings, then a distinguisher will be given strings sampled from either the probability distribution produced by the ideal or the real resource. He then needs to decide from which one the strings were sampled, therefore the distinguishing advantage is equal to the total variation distance between the two probability distributions. If the outputs of the resources are quantum states, then the distringuishing avantage is given by the Helstrom measurement, which depends on the trace distance of the states of the two systems.


\subsection{The Protocol}

In this section we propose a cryptographic protocol that constructs an MPQC resource using quantum and classical communication between the $n$ clients and the Server. We suppose that the clients want to perform a unitary $U$ on their quantum inputs, translated in an MBQC pattern on a brickwork state using measurement angles $\{\phi_j\}_{j=1}^q$, where $q=|O^c|$. For simplicity we also consider that each client $C_k$ ($k\in[n]$) has one qubit as input and one qubit as output, but it is easy to generalise to any other case. We will use the following labeling: client $C_k$ has as input qubit ``$k$'' and as output qubit ``$q+k$'', while the first qubit of the second column in Figure~\ref{brickwork} has label ``$n+1$'', the last one in the second column ``$2n$'' etc.

We want to guarantee that the private data of the clients remain secret during the protocol. Here, each client's  data consists of the quantum input and output, while we consider that the measurement angles are not known to the server (they can be known to all clients, or to a specific client that delegates the computation). The protocol first starts with a process named ``Remote State Preparation'' \cite{DKL12} (see Figures \ref{fig:algo1} and \ref{fig:algo2}). The clients send quantum states to the Server, who then entangles them and measures all but one. In the case where one of the clients has a quantum input, he sends that quantum state one-time padded to the Server, while the rest of the clients send $Z$-rotated $\ket{+}$ states to the Server (Protocol \ref{Algo2}). In the case of the extra ``operational'' qubits in $O^c \setminus I$, all clients send rotated $\ket{+}$ states to the Server (Protocol \ref{Algo3}). In this way, the clients remotely prepare quantum states at the Server's register that are ``encrypted'' using secret data from all of them, without having to communicate quantum states to each other. 

However, since each client is supposed to only choose their own quantum input, and not affect the input of the other clients, the protocol should ask the clients to commit to using the same classical values for the duration of the protocol. This is done by using a VSS scheme each time a classical value is chosen by a client. In the case of the ``Remote State Preparation'' process, each time a client sends a $Z$-rotated $\ket{+}$ state which will need to be corrected at a later point in the protocol. In order to ensure that the ``reverse'' rotation is used later, the clients send many copies of randomly $Z$-rotated $\ket{+}$ states, and commit (via VSS) to the rotations used. They then get tested by the Server and the rest of the clients on the correctness of the committed values. A similar commitment takes place for the quantum one-time pad that each client performs on their quantum input, since the classical values used, affect the measurement angles of consecutive layers of computation.

At the end of the remote state preparation phase, the Server entangles the non-measured states in a universal graph state (for example in the brickwork state of Figure \ref{brickwork} \cite{FK13}). Since the proposed protocol uses MBQC to compute the desired functionality, there is an unavoidable dependency between measurement angles of the qubits in different layers of computation. This means that the clients need to securely communicate between them and with the Server, in order to jointly compute the updated measurement angles, taking into account the necessary corrections from the previous measurements and the dependency sets of each qubit. This procedure is purely classical, and uses VSS schemes and a computation oracle to calculate the necessary values at each step of the protocol and to ensure that the clients behave honestly.

\begin{figure}[h]
\begin{center}
\includegraphics{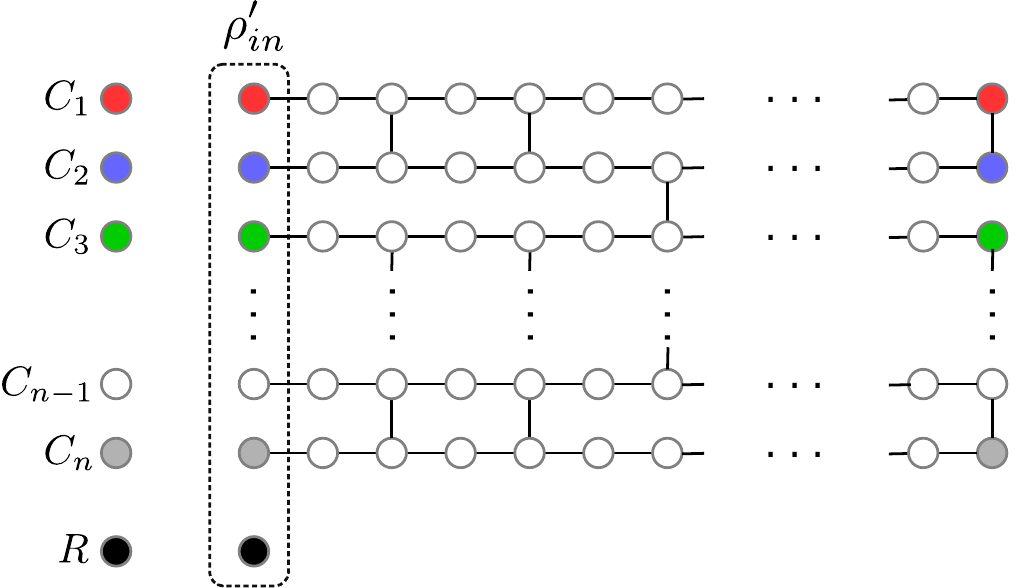}
\end{center}
\caption{The brickwork state with the encrypted quantum input in the first layer. The colors of the qubits denote their origin, while the encrypted input state can also be entangled with the environment $R$. During the computation, all qubits will be measured except the last layer.}\label{brickwork}
\end{figure}

Finally, in the output phase, each output qubit $j\in O$ is naturally encrypted due to the corrections propagated during the computation. Qubit $j$ is sent to the legitimate recipient $C_{j-q}$, while the operation that is needed to decrypt it is the $X^{s^X_j}Z^{s^Z_j}$. The classical values necessary to compute $s^X_j$ and $s^Z_j$ are then computed from the secret shares of all clients and sent to client $C_{j-q}$, who applies the necessary quantum operation.

\begin{algorithm}[H]
\floatname{algorithm}{Protocol}
\caption{(Enforcing honest behavior for client $C_k$)}
\label{Algo1}
\begin{enumerate}
\item Client $C_k$ sends $m$ qubits $\ket{+_{\theta_i^k}}=\frac{1}{\sqrt{2}}(\ket{0}+e^{i\theta_i^k}\ket{1})$ to the Server and secret-shares the values $\{\theta_i^k\}_{i=1}^m$ with all clients, using a VSS scheme.
\item The Server requests the shared values from the clients for all but one qubit, and measures in the resconstructed bases. If the bases agree with the results of the measurements, then with high probability, the remaining state is correctly formed in relation to the shared angle.
\end{enumerate}
\end{algorithm} 


\begin{algorithm}[H]
\floatname{algorithm}{Protocol}
\caption{(State preparation for $j\in I$)}
\label{Algo2}
Server stores states received from clients $C_k$ to distinct registers $\mathcal{S}_k\subset \mathcal{S}$ ($k=1,\dots,n$);

for $k=1,\dots,n-1$

\hspace{0.5in}if $k=j$ then

\hspace{1in}break;

\hspace{0.5in}if $k=n-1$ and $j=n$ then

\hspace{1in}break;

\hspace{0.5in}if $k=j-1$, then

\hspace{1in}CNOT on $\mathcal{S}_k\otimes\mathcal{S}_{k+2}$;

\hspace{0.5in}else

\hspace{1in}CNOT on $\mathcal{S}_k\otimes\mathcal{S}_{k+1}$;

\hspace{0.5in}end;

\hspace{0.5in}measure state in $\mathcal{S}_k$ and get outcome $t_j^k$;

end;

if $j=n$ then 

\hspace{0.5in} CNOT on $\mathcal{S}_{n-1}\otimes\mathcal{S}_n$;

\hspace{0.5in}measure state in $\mathcal{S}_{n-1}$ and get outcome $t_n^{n-1}$;

else

\hspace{0.5in}CNOT on $(\mathcal{S}_{n}\otimes\mathcal{S}_j)$;

\hspace{0.5in}measure state in $\mathcal{S}_n$ and get outcome $t_j^n$;

end;

\end{algorithm}

\begin{figure}[H]

\centerline{
\Qcircuit @C=0.5em @R=1.5em {
\lstick{C_1:~~~~~~~~~~\ket{+_{\theta_j^1}}}	&\qw	&\targ &\meter	&	\cw 	& t_j^1\\	
\lstick{C_2:~~~~~~~~~~\ket{+_{\theta_j^2}}}  	&\qw		&\ctrl{-1}&\qw	&\targ		 &	\meter 	&	\cw		&	t_j^2\\
\lstick{C_3:~~~~~~~~~~\ket{+_{\theta_j^3}}} 	&\qw		&\qw	  &\qw	  	&\ctrl{-1}		&\qw &\targ{+1} & \meter	&	\cw		& \cw	&~~~~  t_j^3\\
\lstick{\vdots \hspace{1in}\vdots}  				&				 &		   &		&		 &			&	\vdots	&		&	& & & & 	&  &  & & \ddots \\
\lstick{C_n:~~~~~~~~~~\ket{+_{\theta_j^n}}}				&\qw   			 &\qw 	   &\qw		& \qw	 &\qw		&\qw		&\qw	&\qw	& \qw	&\qw	&\qw &\qw & \qw &\ctrl{-1}		&\qw &\targ & \meter	&	\cw		& \cw &~~~~~~	t_j^n\\
\lstick{C_j:~X^{a_j}Z(\theta_j^j)\big[\mathcal{C}_j\big] } 	~~~~			&\qw   			 &\qw 	   &\qw		& \qw	 &\qw		&\qw		&\qw	&\qw	& \qw	&\qw	&\qw & \qw & \qw	& \qw	&\qw	& \ctrl{-1}	&\qw & \qw & & &\hspace{0.7in}X^{a_j}Z(\theta_j)\big[\mathcal{C}_j\big] \\\\
					 }		
	}
	\caption{Remote State Preparation with quantum input (Protocol \ref{Algo2}). Client $C_j$ performs a one-time pad on his register $\mathcal{C}_j$ and the result of the circuit remains one-time padded, where $\theta_j=\theta_j^j+\sum_{k=1,k\neq j}^n (-1)^{\bigoplus_{i=k}^n t_j^i+a_j}\theta_j^k$.}\label{fig:algo1}
\end{figure}
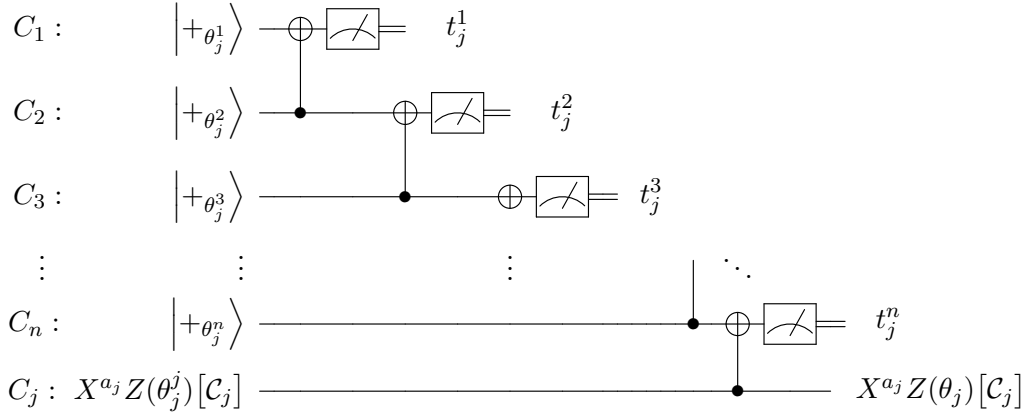

\vspace{0.3in}
\begin{algorithm}[H]
\floatname{algorithm}{Protocol}
\caption{(State preparation for $j\in O^c\setminus I$)}
\label{Algo3}
Server stores states received from clients $C_k$ to distinct registers $\mathcal{S}_k\subset \mathcal{S}$ ($k=1,\dots,n$);

for $k=1,\dots,n-1$

\hspace{0.5in}CNOT on $\mathcal{S}_k\otimes\mathcal{S}_{k+1}$;

\hspace{0.5in}measure state in $\mathcal{S}_k$ and get outcome $t_j^k$;

end;

\end{algorithm} 
\vspace{0.2in}

\begin{figure}[H]
\centerline{
\Qcircuit @C=0.5em @R=1.5em {
\lstick{C_1:~~~~\ket{+_{\theta_j^1}}}	&\qw  &\targ &\meter	&	\cw 	& t_j^1\\	
\lstick{C_2:~~~~\ket{+_{\theta_j^2}}}  	&\qw		&\ctrl{-1}&\qw	&\targ		 &	\meter 	&	\cw		&	t_j^2\\
\lstick{C_3:~~~~\ket{+_{\theta_j^3}}} 	&\qw		&\qw	  &\qw	  	&\ctrl{-1}		&\qw &\targ & \meter	&	\cw		& \cw	&~~~~  t_j^3\\
\lstick{\vdots \hspace{0.7in}\vdots}  \hspace{0.5in}				&				 &		   &		&		 &			&	\vdots	&		&	& & & & 	&  &  & & \ddots \\
\lstick{C_{n-1}:\hspace{0.1in}\ket{+_{\theta_j^{n-1}}}}				&\qw   			 &\qw 	   &\qw		& \qw	 &\qw		&\qw		&\qw	&\qw	& \qw	&\qw	&\qw &\qw & \qw &\ctrl{-1}		&\qw &\targ & \meter	&	\cw		& \cw & ~~~~~	t_j^{n-1}\\
\lstick{C_n:~~~~\ket{+_{\theta_j^n}}}				&\qw   			 &\qw 	   &\qw		& \qw	 &\qw		&\qw		&\qw	&\qw	& \qw	&\qw	&\qw & \qw & \qw	& \qw	&\qw	& \ctrl{-1}	&\qw & \qw & & &\hspace{0.3in}\mathbf{\ket{+_{\theta_j}}}\\
					 }		
	}
	\caption{Remote State Preparation without quantum input (Protocol \ref{Algo3}), where $\theta_j=\theta_j^n+\sum_{k=1}^{n-1} (-1)^{\bigoplus_{i=k}^{n-1} t_j^i}\theta_j^k$.}\label{fig:algo2}
\end{figure}
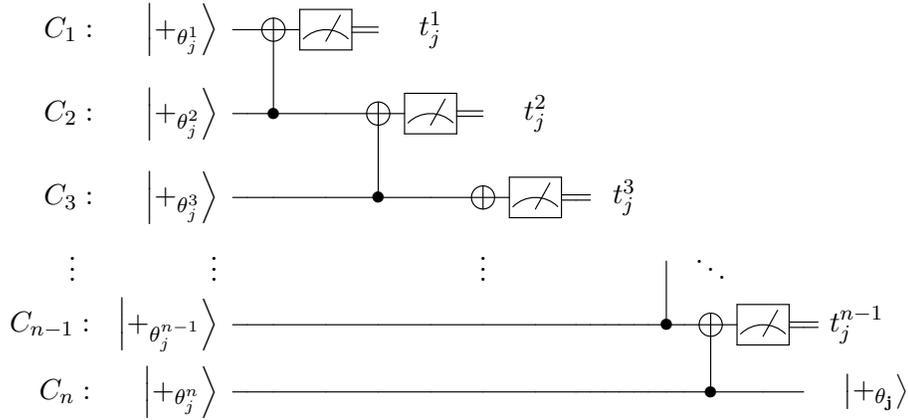

\begin{algorithm}[!h]
\caption{Multiparty Quantum Computing Protocol}
\label{protocol}
\begin{itemize}
\item A quantum input $\rho_{in}$ and measurement angles $\{\phi_j\}_{j=1}^q$ for qubits $j\in O^c$ .
\end{itemize}

\underline{\emph{Preparation phase}}
\begin{description}
\item[quantum input:] For $j\in I$ 
\begin{enumerate}
\item Client $C_j$ applies a one-time pad $X^{a_j}Z(\theta_j^j)$ to his qubit, where $a_j\in_R\{0,1\}$ and $\theta_j^j\in_R\{l\pi/4\}_{l=0}^7$ and sends it to the Server. He secret-shares the values $a_j$ and $\theta_j^j$ with the other clients.
\item Each client $C_k (k\neq j)$, runs Protocol \ref{Algo1} with the Server. If all clients pass the test, the Server at the end has $n-1$ states $\ket{+_{\theta_j^k}}=\frac{1}{\sqrt{2}}\big(\ket{0}+e^{i\theta_j^k}\ket{1}  \big)$ for $k\neq j$.
\item The Server runs Protocol \ref{Algo2} and announces outcome vector $\mathbf{t}_j$. \\
\end{enumerate}\vspace{-0.1in}
At this point the Server has the state  $\rho'_{in}=\big(X^{a_1}Z(\theta_1)\otimes \dots \otimes X^{a_n} Z(\theta_n)\otimes \mathbf{1}_{\mathcal{R}}\big)\cdot \rho_{in}$, where 
 \begin{equation}\label{eq:entangle1}
 \theta_j=\theta_j^j+\sum_{k=1, k\neq j}^n (-1)^{\bigoplus_{i=k}^n t_j^i+a_j}\theta_j^k
\end{equation}
\item[non-output / non-input qubits:] For $j\in O^c\setminus I$ 
\begin{enumerate}
\item[4.] All clients $C_k$, $k\in[n]$  run Protocol \ref{Algo1} with the Server. If all clients pass the test, the Server at the end has $n$ states $\ket{+_{\theta_j^k}}=\frac{1}{\sqrt{2}}\big(\ket{0}+e^{i\theta_j^k}\ket{1}  \big)$ for $k=1,\dots,n$.
\item[5.] The Server runs Protocol \ref{Algo3} getting outcome vector $\mathbf{t}_j$. He ends up with the state $\ket{+_{\theta_j}}$, where:
\begin{equation}\label{eq:entangle2}
\theta_j=\theta_j^n+\sum_{k=1}^{n-1} (-1)^{\bigoplus_{i=k}^{n-1} t_j^i}\theta_j^k
\end{equation}
\end{enumerate}
\item[output qubits:] For $j\in O$, the Server prepares $\ket{+}$ states.
\item[graph state:] The Server entangles the $n+q$ qubits to a brickwork state by applying ctrl-$Z$ gates.

\end{description}

\begin{flushleft}
\underline{\emph{Computation phase}}
\vspace{-7pt}
\end{flushleft}

\begin{description}
\item[non-output qubits:] For $j\in O^c$

\begin{enumerate}

\item All clients $C_k$, $k=1,\dots,n$ choose random $r_j^k\in\{0,1\}$, which they secret-share with the other clients. Then using a computation oracle, they compute the measurement angle of qubit $j$:
\begin{equation}\label{angle}
\delta_j:=\phi'_j+\pi r_j+\theta_j
\end{equation}
where undefined values are equal to zero, or otherwise:
\begin{itemize}
\item $\phi'_j=(-1)^{a_j+s_j^X}\phi_j+s^Z_j\pi+a_{f^{-1}(j)}\pi$.
\item $r_j=\bigoplus\limits_{k=1}^n r_j^k$.
\item $s_i=b_i\oplus r_i$, for $i\leq j$.
\end{itemize}

\item The Server receives $\delta_j$ and measures qubit $j$ in basis $\{\ket{+_{\delta_j}},\ket{-_{\delta_j}}\}$, getting result $b_j$. He announces $b_j$ to the clients.
\end{enumerate}
\item[output qubits:] For $j\in O$, the Server sends the ``encrypted'' quantum state to client $C_{j-q}$. All participants jointly compute $s_j^X$ and $s_j^Z$ and send it to client $C_{j-q}$, who applies operation $Z^{s_j^Z}X^{s_j^X}$ to retrieve the actual quantum output. 

\end{description}
\end{algorithm}

\clearpage

\subsection{Analysis of the Protocol}
\subsubsection{Correctness}

\begin{theorem}\label{thm:correctness}
Protocol \ref{protocol} emulates the filtered Ideal Resource $\mathcal{M}_\perp$ of Figure \ref{fig:MPQC}.
\end{theorem}

The vailidity of the above theorem comes directly from the correctness of the individual circuits implementing Protocols \ref{Algo2} and \ref{Algo3}, as well the propagation of $Z$ and $X$ corrections through the flow of the computation. A detailed proof is given in the Appendix that shows that:
\[d(\pi \mathcal{O}, \mathcal{M}_\perp)=0\]
therefore a distinguisher cannot tell the difference between the real communication protocol and an interaction with the Ideal MPQC Resource when all participants are honest.

\subsubsection{Malicious Server}
The proof of security against a malicious Server that is allowed to deviate from the protocol, by applying operations on the data he receives, is based on quantum teleportation. As observed in \cite{DFPR14}, for all the quantum states sent to the Server in Protocol \ref{protocol}, there exists an equivalent circuit that uses an EPR pair and teleports via measurement the quantum state to one of the entangled qubits of the EPR, which is then sent to the Server. What is more important, is that the operation on the half of the EPR pair, can happen at the Server side, before the ``teleportation'' of information from the clients. We can therefore build a simulator for the Server $\sigma_S$, who runs Protocol \ref{prot:sim}. We can also define an MPQC resource that  runs Protocol \ref{prot:mpqc} when it interacts with the simulator on behalf of the Server.

\begin{algorithm}[!h]
\caption{Simulator for Server}
\label{prot:sim}
\begin{description}
\item[non-output qubits:] For $j\in O^c$
\begin{enumerate}
\item $\sigma_S$ creates an EPR pair and sends one half of it to the Server.
\item $\sigma_S$ runs Protocol \ref{Algo1} on behalf of the clients $C_k,~k\neq j$ when $j\in I$ and of the clients $C_k,~k\in[n-1]$ for $j\in O^c\setminus I$ sending to the Server half EPR pairs and always accepting.
\item $\sigma_S$ receives vector $\mathbf{t}_j$. 
\item $\sigma_S$ sends $\delta_j\in_R\{l\pi/4\}_{l=1}^7$ to the Server and receives a reply $b_j\in\{0,1\}$.
\end{enumerate}
\item[output qubits:] For $j\in O$
\begin{enumerate}
\item $\sigma_S$ receives $n$ qubits from the Server.
\item $\sigma_S$ sends the other halfs of the EPR pairs, the received quantum states, as well as $\delta_j$, $b_j$ and $\mathbf{t}_j$ for $j=1,\dots,q$, to the MPQC resource.
\end{enumerate}
\end{description}
\end{algorithm}

\begin{algorithm}[!h]
\caption{MPQC resource}
\label{prot:mpqc}
\begin{enumerate}
\item The resource receives the $n$ qubits of $\rho_{in}$ from the clients and all the information from $\sigma_S$.
\item For $j\in I$: the resource performs a CNOT on the corresponding EPR half with the input qubit as control and measures the EPR half in the computational basis, getting result $a_j$. It chooses random measurement angles $\hat{\theta}_j^k$ for the qubits coming from clients $C_k,k\neq j$, sets:
\begin{equation*}
\hat{\theta}_j^j:=\delta_j-\phi'_j-\sum_{k=1, k\neq j}^n (-1)^{\bigoplus_{i=k}^n t_j^i+a_j}\hat{\theta}_j^k
\end{equation*} 
and measures the corresponding qubits. For $j\in O^c\setminus I$, it chooses random measurement angles $\hat{\theta}_j^k$ for clients $C_k,k\in[n-1]$,sets:
\begin{equation*}
\hat{\theta}_j^n:=\delta_j-\phi'_j-\sum_{k=1}^{n-1} (-1)^{\bigoplus_{i=k}^{n-1} t_j^i}\hat{\theta}_j^k
\end{equation*}
and measures the corresponding qubits. In the computation of the angles, undefined values are equal to zero, and:
\begin{itemize}
\item $\phi'_j=(-1)^{a_j+s_j^X}\phi_j+s^Z_j\pi+a_{f^{-1}(j)}\pi$.
\item $r_j=\bigoplus\limits_{k=1}^n r_j^k$.
\item $s_i=b_i\oplus r_i$, for $i\leq j$.
\end{itemize}
\item For $j\in O$: the resource performs corrections $Z^{s_j^Z}X^{s_j^X}$ on the remaining qubits.
\end{enumerate}
\end{algorithm}

\begin{theorem}
Protocol \ref{protocol} is secure against a malicious Server.
\end{theorem}
\begin{proof}
In order to show that the protocol is secure against a malicious Server, we will argue that Eq.(\ref{secure1}) holds for the simulator and the MPQC resource presented below. The proof relies heavily on teleportation techniques and more specifically in the equivalence of circuits and delay of operations on different wires. We want to be sure that Protocol \ref{protocol} and the MPQC resource in Protocol \ref{prot:mpqc} implement the same map, in other words that the outputs are indistinguishable. The equivalence comes from the fact that the one-time pad on the clients' input (in Protocol \ref{protocol}) can be rewritten as a ``delayed'' teleportation of the input to the EPR halfs that the Server acts upon. Since the simulator does not have access to the quantum input or the measurement angles, it then sends the other half of the EPR pair back to the Ideal Resource, so that it can get entangled with the quantum input and measured in the correct basis.

Similarly for the rotated $\ket{+_{\theta_j}}$ states, since the choice of the rotations is done uniformly at random, but defines the measurement, we can again think of an equivalent ``delayed'' teleportation, where the measurement happens at random, and then the rotation is chosen. The random values $r_j$ that are chosen at random in the protocol, also follow the uniform distribution in the measurements done by the MPQC resource, due to the properties of the entanglement in the teleportation process.

Instead of choosing $\theta_j$, $a_j$ and $r_j$ uniformly at random from their respective domains at the different steps of the protocol, the simulator chooses $\delta_j$ uniformly at random and the two bits $a_j$ and $r_j$ get defined from the randomness inherent when measuring the EPR pairs. Therefore, the random variables that are chosen during the two protocols are equivalent, since in Protocol \ref{protocol} the choice of the variables define $\delta_j$ while in Protocol \ref{prot:sim}, the simulator chooses $\delta_j$ uniformly at random and then the value of $\theta'_j$ is defined. However, in the simulated protocol, the Server has at no point access to the quantum input $\rho_{in}$ or the measurement angles, therefore the security of the simulated protocol and by consequence the security of the Real protocol is ensured (see Appendix for more details).

\end{proof}

We have proven that a Server does not learn anything about the inputs of the clients, since the protocol emulates an ideal MPQC resource that does not provide any access to the clients' input, to the Server at any point. From this result, we prove at the same time that the computation is done in a blind way, meaning that the Server does not know what computation he is performing, as long as the measurement pattern (i.e. the angles $\phi_j$) remains hidden from him. The proof of blindness follows directly from the equivalence of protocols, since again the Server has no access to the measurement angles at any point, when interacting with the ideal MPQC resource.

\begin{corollary}
The protocol is blind against a malicious Server.
\end{corollary}

\subsubsection{Malicious Clients}

\begin{theorem}The protocol is secure against a coalition of malicious clients.
\end{theorem}

\begin{proof}
Protocol \ref{prot:sim2} presents a simulator $\sigma_C$ that receives communication from a malicious coalition of clients on its external interface. For ease of use we will consider one malicious client $C_c$ with one input qubit, but this can easily be extended by thinking of all malicious clients as one client that has multiple input qubits. It is straightforward to see that Equation \eqref{secure2} holds, since the malicious clients never receive quantum information from the other clients, and the only information they share is the one used in the computation oracles that are implemented using secure classical multiparty computation protocols. The quantum outcome they receive is the correct outcome of an honest Server, encoded by some information ($r_j$) that is chosen by the malicious clients in a previous step.

\begin{algorithm}[!h]
\caption{Simulator for Clients}
\label{prot:sim2}
\begin{enumerate}
\item $\sigma_C$ receives a quantum state from client $C_c$ as well as the secret shares of $a_j$ and $\theta_j^j$.
\item For all other nodes of the brickwork state, $\sigma_C$ runs Protocol \ref{Algo1} with Client $C_c$ and aborts if the secret shares of the classical values do not all match the measurement outcomes of the quantum states.
\item For $j\in O^c$, $\sigma_C$ receives the secret shares of the randomness $r_j^c$, chosen by the Client $C_c$ and interacts according to the communication protocol simulating the oracle of computing $\delta_j$, choosing uniformly at random the value of $\delta_j$. $\sigma_C$ also replies with random $b_j$.
\item $\sigma_C$ undoes rotation $X^{a_c}Z(\theta_c^c)$ on the input qubit of $C_c$, inputs it to the Ideal MPQC resource and gets back the output corresponding to $C_c$.
\item Finally, $\sigma_C$ rotates the output qubit $j\in O$ corresponding to client $C_c$, applying the operation $Z^{s_j^Z}X^{s_j^X}$, and participates in the computation protocol to compute $s_j^X$ and $s_j^Z$ with the previously sent and shared values $b_j$ and $r_j$. 

\end{enumerate}
\end{algorithm}

\end{proof}


\section{Conclusion}
In this work, we have presented a quantum multiparty delegated protocol that provides security for clients with limited quantum abilities, therefore  extending previous results on two-party \cite{DNS10} and multiparty \cite{Benor} computation, with recent work on delegated blind computing \cite{FK13, DFPR14, UBQC}. Our protocol requires no quantum memory for the clients and no need for entangling operations or measurement devices. The only operations that they need to do in the most general case of quantum input and output, is X gates and Z rotations. Our protocol is secure against a dishonest Server, or a coalition of malicious clients. It remains to study whether the proposed protocol remains secure against a dishonest coalition between clients and the Server or if there is an unavoidable leakage of information. One equivalent way of studying this problem would be by extending the results of~\cite{Wallden16} in the multiparty setting, where both the parties and the Server have inputs in the computation. An even more interesting question is whether we can enhance our protocol to include verifiability in a similar way that is done in \cite{FK13}.

The specific protocol presented here is a multiparty version of the delegated blind protocol of \cite{UBQC} and as such, inherits the key advantage of using MBQC over gate teleportation approaches; once the preparation phase is finished (all qubits are sent to the Server and entangled in a graph state), the rest of the communication is classical. However, it can easily be adapted to any blind computing model, for example the measurement-only model \cite{MF13}, since as mentioned in \cite{DFPR14}, all protocols with one-way communication from the Server to a client, are inherently secure due to no-signaling. We have also assumed that the clients choose to act passively malicious, since any active dishonest activity would be detected with very high probability; however, a quantitative proof of security, assuming more extensive attacks from the side of the clients would be an natural extention of this work.

Finally, a similar approach to ours has been explored for two-party computation, that uses recent advances in classical Fully Homomorphic encryption (FHE). In \cite{BJ} and in follow-up work \cite{DSS16} it is shown how to evaluate quantum circuits using quantum FHE, and it would be very interesting to see how they can be adapted in the case of multiple parties and whether the computational and communication requirements are different from our work.   

\vspace{0.4in}

\noindent{\bf Acknowledgements:} This work was supported by grants EP/N003829/1 and EP/M013243/1 from the UK Engineering and Physical Sciences Research Council (EPSRC) and by the Marie Sklodowska-Curie Grant Agreement No. 705194 from the European Union's Horizon 2020 Research and Innovation program. Part of this work was done while AP was visiting LTCI, Paris sponsored by a mobility grant from SICSA. The authors would also like to thank Petros Wallden and Theodoros Kapourniotis for many useful discussions.

\clearpage

\section{Appendix}
We first give the complete proof of Theorem \ref{thm:correctness}.

\begin{proof}
We want to prove that Equation \eqref{correct} holds.
First, we will prove that the Equations (\ref{eq:entangle1}) and (\ref{eq:entangle2}) are correct. We will start with Eq.(\ref{eq:entangle2}) for a server that receives $n$ qubits $\ket{+_{\theta_j^k}}$ for $k=1,\dots,n$.

We will show the result by induction to the number of qubits received. For the case of two qubits $\ket{+_{\theta_j^1}}$ and  $\ket{+_{\theta_j^2}}$, the Server performs a CNOT operation on them (with control wire the second one). The resulting state is:
\[
\ket{0}\ket{+_{\theta_j^2+\theta_j^1}}+e^{i\theta_j^1}\ket{1}\ket{+_{\theta_j^2-\theta_j^1}}
\]
When the Server measures the first qubit, he sets outcome bit $t_j^1=0$ when the observed state is $\ket{0}$ and  $t_j^1=1$ when it is $\ket{1}$. Therefore the resulting state is $\ket{+_{\tilde{\theta}_j^2}},$ where:
\[ \tilde{\theta}_j^2=\theta_j^2+(-1)^{t_j^1}\theta_j^1
\]
So Eq.(\ref{eq:entangle2}) holds for $n=2$. Now we will assume that the claim holds for $n-1$ and we will prove it for $n$. After measurement of the qubit $n-2$, the state of qubit $n-1$ is $\ket{+_{\tilde{\theta}_j^{n-1}}}$, where:
\[ \tilde{\theta}_j^{n-1}=\theta_j^{n-1}+\sum_{k=1}^{n-2} (-1)^{\bigoplus_{i=k}^{n-2} t_j^i}\theta_j^k
\]
The Server performs a CNOT on qubits $n-1$ and $n$, resulting in the state:
\[
\ket{0}\ket{+_{\theta_j^n+\tilde{\theta}_j^{n-1}}}+e^{i\tilde{\theta}_j^{n-1}}\ket{1}\ket{+_{\theta_j^n-\tilde{\theta}_j^{n-1}}}
\]
The state after the measurement of the qubit $n-1$ is $\ket{+_{\theta_j}}$, where:
\begin{eqnarray*}
\theta_j&=&\theta_j^n+(-1)^{t_j^{n-1}}\tilde{\theta}_j^{n-1}\\
             &=&\theta_j^n+(-1)^{t_j^{n-1}}\big(\theta_j^{n-1}+\sum_{k=1}^{n-2} (-1)^{\bigoplus_{i=k}^{n-2} t_j^i}\theta_j^k\big ) \\
             &=&\theta_j^n+\sum_{k=1}^{n-1} (-1)^{\bigoplus_{i=k}^{n-1} t_j^i}\theta_j^k
\end{eqnarray*}

\noindent{}We have therefore proven Eq.(\ref{eq:entangle2}). What remains in order to prove Eq.(\ref{eq:entangle1}) is to see what happens when the Server entangles the one-padded quantum input of client $C_j$ with the states $\ket{+_{\theta_j^k}}$ of the rest of the clients $C_k,~ k\neq j$. If the Server follows subprotocol \ref{Algo2}, he first entangles the rotated qubits of the clients $C_k,~ k\neq j$ and measures all but the last, creating a state $\ket{+_{\tilde{\theta}_j}}$. Now we know how to compute $\tilde{\theta}_j$:
\begin{itemize}
\item  For $j\neq n$: $\tilde{\theta}_j=\theta_j^n+\sum_{k=1,k\neq j}^{n-1} (-1)^{\bigoplus_{i=k}^{n-1} t_j^i}\theta_j^k$. 
\item For $j=n$: $\tilde{\theta}_j=\theta_n^{n-1}+\sum_{k=1}^{n-2} (-1)^{\bigoplus_{i=k}^{n-2} t_n^i}\theta_n^k$. 
\end{itemize}

\noindent The last step of Protocol \ref{Algo2} performs a CNOT on $\ket{+_{\tilde{\theta}_j}}$ with control qubit the one-time padded input of client $C_j$ and measures the first in the computational basis. We already have seen how the Z-rotation propagates through the CNOT gate. The $X$ operation of the one-time pad results in a bit flip of the last measurement outcome (either $n$ or $n-1$ according to the two cases above). Therefore, if the one-time pad on register $\mathcal{C}_j$ was $X^{a_j}Z(\theta_j^j)$, after the Remote State Preparation, the register $\mathcal{C}_j$ is still one-time padded with $X^{a_j}Z(\theta_j)$, where:
\begin{itemize}
\item  For $j\neq n$: 
\begin{equation*}
\theta_j=\theta_j^j+(-1)^{t_j^n+a_j}\Big(\theta_j^n+\sum_{k=1,k\neq j}^{n-1} (-1)^{\bigoplus_{i=k}^{n-1} t_j^i}\theta_j^k\Big)=\theta_j^j+\sum_{k=1,k\neq j}^n (-1)^{\bigoplus_{i=k}^n t_j^i+a_j}\theta_j^k
\end{equation*}
\item For $j=n$:
\begin{equation*}
\theta_n=\theta_n^n+(-1)^{t_n^{n-1}+a_n}\Big(\theta_n^{n-1}+\sum_{k=1}^{n-2} (-1)^{\bigoplus_{i=k}^{n-2} t_n^i}\theta_n^k\Big)=\theta_n^n+\sum_{k=1}^{n-1} (-1)^{\bigoplus_{i=k}^{n-1} t_n^i+a_n}\theta_n^k
\end{equation*}
\end{itemize}

From the two cases above, it is obvious that for a general $j=1,\dots,n$, Eq.(\ref{eq:entangle1}) is true, which concludes the correctness of the preparation phase of Protocol \ref{protocol}.

In the Computation phase, for each qubit of the brickwork state, all clients input their data in the classical box and the output is the measurement angle of that qubit. From the properties of MBQC and the flow of the protocol, when the entangled qubits in the Preparation phase are in the $\ket{+}$ state, each angle $\phi_j$ needs to be adjusted to $\phi_j'$ as defined above. However, here each qubit is rotated by all clients by a total angle $\theta_j$. Therefore the measurement angle in the MBQC needs to be adjusted by $\theta_j$. The final difference from ordinary MBQC is the insertion of some joint randomness $r_j$, whose effect is reversed in the consequent steps by adding the randomness to the correction of the function (see \cite{UBQC} for details).

Finally, in the outcome phase, due to previous corrections, the state of qubit $j\in O$ needs to be corrected by client $C_{j-q}$ by applying an operation $Z^{s_j^Z}X^{s_j^X}$, whose classical values can be computed using a computation oracle. We have therefore proven that:
\[d(\pi \mathcal{O}, \mathcal{M}_\perp)=0\]
meaning that a distinguisher $\mathcal{Z}$ cannot tell the difference between the real communication protocol and an interaction with the Ideal MPQC Resource when all participants are honest.

\end{proof}
\vspace{0.4in}

Now, we will present intermediate protocols that prove Equation \eqref{secure1} similarly to the proof technique used in \cite{DFPR14}. We will not include any test of correctness for the clients (Protocol \ref{Algo1} and secret sharing schemes) since the technique used is based on teleportation and delayed measurements and therefore it is not possible for the clients to commit to the correct preparation of the quantum states beforehand. However, this does not affect the proof of security, since these protocols are artificial and used only to show that a malicious Server does not have access at any step to the private data of the clients. We could have included these tests of correctness of the clients, always asking them to accept any measurement outcome of the Server, and therefore showing that they do not provide any further information to the Server. The complete real communication protocol and the simulated one are presented in the main text. Here, we restate the communication protocol, omitting the steps where the clients' honest behavior is checked (Protocol \ref{app:protocol1}) and provide intermediate protocols that are used to prove the equivalence of the real and ideal setting in the case of a malicious Server. This will be done by a step-wise process of proving that each of the presented protocols is equivalent to the others, leading to the final one that uses a simulator for the Server and the Ideal Resource defined in the main text.

\begin{algorithm}[h]
\caption{Multiparty Quantum Computing}
\label{app:protocol1}

\begin{itemize}
\item A quantum input $\rho_{in}$ and measurement angles $\{\phi_j\}_{j=1}^q$ for qubits $j\in O^c$ .
\end{itemize}

\underline{\emph{Preparation phase}}
\begin{description}
\item[quantum input:] For $j\in I$ 
\begin{enumerate}
\item \label{1.1} Client $C_j$ applies a one-time pad $X^{a_j}Z(\theta_j^j)$ to his qubit, where $a_j\in_R\{0,1\}$ and $\theta_j^j\in_R\{l\pi/4\}_{l=0}^7$ and sends it to the Server. 
\item \label{1.2} Each client $C_k~ (k\neq j)$ chooses $\theta_j^k\in_R\{l\pi/4\}_{l=0}^7$ and sends $\ket{+_{\theta_j^k}}=\frac{1}{\sqrt{2}}\big(\ket{0}+e^{i\theta_j^k}\ket{1}  \big)$ to the Server.
\item The Server runs Protocol \ref{Algo2} and announces outcome vector $\mathbf{t}_j$. \\
\end{enumerate}\vspace{-0.1in}
At this point the Server has the state  $\rho'_{in}=\big(X^{a_1}Z(\theta_1)\otimes \dots \otimes X^{a_n} Z(\theta_n)\otimes \mathbf{1}_{\mathcal{R}}\big)\cdot \rho_{in}$, where 
 \begin{equation}\label{appeq:entangle1.1}
 \theta_j=\theta_j^j+\sum_{k=1, k\neq j}^n (-1)^{\bigoplus_{i=k}^n t_j^i+a_j}\theta_j^k
\end{equation}
\item[non-output / non-input qubits:] For $j\in O^c\setminus I$ 
\begin{enumerate}
\item[4.] \label{1.4} All clients $C_k$, $k\in[n]$ choose $\theta_j^k\in_R\{l\pi/4\}_{l=0}^7$ and send $\ket{+_{\theta_j^k}}=\frac{1}{\sqrt{2}}\big(\ket{0}+e^{i\theta_j^k}\ket{1}  \big)$ to the Server.
\item[5.] \label{1.5} The Server runs Protocol \ref{Algo3} getting outcome vector $\mathbf{t}_j$. He ends up with the state $\ket{+_{\theta_j}}$, where:
\begin{equation}\label{appeq:entangle1.2}
\theta_j=\theta_j^n+\sum_{k=1}^{n-1} (-1)^{\bigoplus_{i=k}^{n-1} t_j^i}\theta_j^k
\end{equation}
\end{enumerate}
\item[output qubits:] For $j\in O$, the Server prepares $\ket{+}$ states.
\item[graph state:] The Server entangles the $n+q$ qubits to a brickwork state by applying ctrl-$Z$ gates.

\end{description}

\begin{flushleft}
\underline{\emph{Computation phase}}
\vspace{-7pt}
\end{flushleft}

\begin{description}
\item[non-output qubits:] For $j\in O^c$

\begin{enumerate}

\item All clients $C_k$, $k=1,\dots,n$ choose random $r_j^k\in\{0,1\}$ and using a computation oracle, they compute the measurement angle of qubit $j$:
\begin{equation}\label{angle}
\delta_j:=\phi'_j+\pi r_j+\theta_j
\end{equation}
where undefined values are equal to zero, or otherwise:
\begin{itemize}
\item $\phi'_j=(-1)^{a_j+s_j^X}\phi_j+s^Z_j\pi+a_{f^{-1}(j)}\pi$.
\item $r_j=\bigoplus\limits_{k=1}^n r_j^k$.
\item $s_i=b_i\oplus r_i$, for $i\leq j$.
\end{itemize}

\item The Server receives $\delta_j$ and measures qubit $j$ in basis $\{\ket{+_{\delta_j}},\ket{-_{\delta_j}}\}$, getting result $b_j$. He announces $b_j$ to the clients.
\end{enumerate}
\item[output qubits:] For $j\in O$, the Server sends the ``encrypted'' quantum state to client $C_{j-q}$. All participants jointly compute $s_j^X$ and $s_j^Z$ and send it to client $C_{j-q}$, who applies operation $Z^{s_j^Z}X^{s_j^X}$ to retrieve the actual quantum output. 

\end{description}
\end{algorithm}

\begin{algorithm}[h]
\caption{Multiparty Quantum Computing (equivalent version one)}
\label{app:protocol2}

\begin{itemize}
\item A quantum input $\rho_{in}$ and measurement angles $\{\phi_j\}_{j=1}^q$ for qubits $j\in O^c$ .
\end{itemize}

\underline{\emph{Preparation phase}}
\begin{description}
\item[quantum input:] For $j\in I$ 
\begin{enumerate}
\item \label{2.1} Client $C_j$ creates an EPR pair $\frac{1}{\sqrt{2}}(\ket{00}+\ket{11})$ and sends half to the Server. He then applies a $Z(\hat{\theta}_j^j)$ rotation to his qubit, where $\hat{\theta}_j^j\in_R\{l\pi/4\}_{l=0}^7$, performs a CNOT on the remaining half EPR qubit with control the input qubit, and  measures the input qubit in the Hadamard basis and the half EPR in the computational basis, getting outcomes $r_j^j$ and $a_j$ respectively.
\item \label{2.2} Each client $C_k~ (k\neq j)$ creates an EPR pair $\frac{1}{\sqrt{2}}(\ket{00}+\ket{11})$ and sends half to the Server. He then chooses $\hat{\theta}_j^k\in_R\{l\pi/4\}_{l=0}^7$ and applies a $Z(\hat{\theta}_j^k)$ rotation to the remaining half EPR and then measures it in the Hadamard basis getting outcome $r_j^k$.
\item The Server runs Protocol \ref{Algo2} and announces outcome vector $\mathbf{t}_j$. \\
\end{enumerate}\vspace{-0.1in}
At this point the Server has the state  $\rho'_{in}=\big(X^{a_1}Z(\theta_1)\otimes \dots \otimes X^{a_n} Z(\theta_n)\otimes \mathbf{1}_{\mathcal{R}}\big)\cdot \rho_{in}$, where 
 \begin{equation}\label{appeq:entangle2.1}
 \theta_j=\pi\bigoplus_{k=1}^n r_j^k+\hat{\theta}_j^j+\sum_{k=1, k\neq j}^n (-1)^{\bigoplus_{i=k}^n t_j^i+a_j}\hat{\theta}_j^k
\end{equation}
\item[non-output / non-input qubits:] For $j\in O^c\setminus I$ 
\begin{enumerate}
\item[4.] Each client $C_k$, $k\in[n]$  creates an EPR pair $\frac{1}{\sqrt{2}}(\ket{00}+\ket{11})$ and sends half to the Server. He then chooses $\hat{\theta}_j^k\in_R\{l\pi/4\}_{l=0}^7$ and applies a $Z(\hat{\theta}_j^k)$ rotation to the remaining half EPR followed by a Hadamard, and then measures it in the computational basis getting outcome $r_j^k$.
\item[5.] The Server runs Protocol \ref{Algo3} getting outcome vector $\mathbf{t}_j$. He ends up with the state $\ket{+_{\theta_j}}$, where:
\begin{equation}\label{appeq:entangle2.2}
\theta_j=\pi\bigoplus_{k=1}^n r_j^k+\hat{\theta}_j^n+\sum_{k=1}^{n-1} (-1)^{\bigoplus_{i=k}^{n-1} t_j^i}\hat{\theta}_j^k
\end{equation}
\end{enumerate}
\item[output qubits:] For $j\in O$, the Server prepares $\ket{+}$ states.
\item[graph state:] The Server entangles the $n+q$ qubits to a brickwork state by applying ctrl-$Z$ gates.

\end{description}

\begin{flushleft}
\underline{\emph{Computation phase}}
\vspace{-7pt}
\end{flushleft}

\begin{description}
\item[non-output qubits:] For $j\in O^c$

\begin{enumerate}

\item The clients use a computation oracle to send the measurement angle of qubit $j$ to the Server:
\begin{equation}\label{angle}
\delta_j:=\phi'_j+\pi r_j+\theta_j
\end{equation}
where undefined values are equal to zero, or otherwise:
\begin{itemize}
\item $\phi'_j=(-1)^{a_j+s_j^X}\phi_j+s^Z_j\pi+a_{f^{-1}(j)}\pi$.
\item $r_j=\bigoplus\limits_{k=1}^n r_j^k$.
\item $s_i=b_i\oplus r_i$, for $i\leq j$.
\end{itemize}

\item The Server measures qubit $j$ in basis $\{\ket{+_{\delta_j}},\ket{-_{\delta_j}}\}$ and announces result $b_j$. 
\end{enumerate}
\item[output qubits:] For $j\in O$, the Server sends the ``encrypted'' quantum state to client $C_{j-q}$. All participants jointly compute $s_j^X$ and $s_j^Z$ and send it to client $C_{j-q}$, who applies operation $Z^{s_j^Z}X^{s_j^X}$ to retrieve the actual quantum output. 

\end{description}
\end{algorithm}

\begin{algorithm}[h]
\caption{Multiparty Quantum Computing (equivalent version two)}
\label{app:protocol3}

\begin{itemize}
\item A quantum input $\rho_{in}$ and measurement angles $\{\phi_j\}_{j=1}^q$ for qubits $j\in O^c$ .
\end{itemize}

\underline{\emph{Preparation phase}}
\begin{description}
\item[quantum input:] For $j\in I$ 
\begin{enumerate}
\item Client $C_j$ creates an EPR pair $\frac{1}{\sqrt{2}}(\ket{00}+\ket{11})$ and sends half to the Server. He then performs a CNOT on the remaining half EPR qubit with control the input qubit, and  measures the former in the computational basis, getting outcome $a_j$.
\item Each client $C_k~ (k\neq j)$ creates an EPR pair $\frac{1}{\sqrt{2}}(\ket{00}+\ket{11})$ and sends half to the Server. 
\item The Server runs Protocol \ref{Algo2} and announces outcome vector $\mathbf{t}_j$. \\
\end{enumerate}\vspace{-0.1in}
\item[non-output / non-input qubits:] For $j\in O^c\setminus I$ 
\begin{enumerate}
\item[4.] Each client $C_k$, $k\in[n]$  creates an EPR pair $\frac{1}{\sqrt{2}}(\ket{00}+\ket{11})$ and sends half to the Server. 
\item[5.] The Server runs Protocol \ref{Algo3} getting outcome vector $\mathbf{t}_j$. 
\end{enumerate}
\item[output qubits:] For $j\in O$, the Server prepares $\ket{+}$ states.
\item[graph state:] The Server entangles the $n+q$ qubits to a brickwork state by applying ctrl-$Z$ gates.

\end{description}

\begin{flushleft}
\underline{\emph{Computation phase}}
\vspace{-7pt}
\end{flushleft}

\begin{description}
\item[non-output qubits:] For $j\in O^c$

\begin{enumerate}

\item The computation oracle sends a random angle $\delta_j\in_R \{l\pi/4\}_{l=0}^7$ to the Server, who measures qubit $j$ in basis $\{\ket{+_{\delta_j}},\ket{-_{\delta_j}}\}$ and announces result $b_j$. 

\item For $j\in I$, the computation oracle chooses random measurement angles $\hat{\theta}_j^k$ for clients $C_k,k\neq j$ and sets:
\begin{equation}\label{app:angle}
\hat{\theta}_j^j:=\delta_j-\phi'_j-\sum_{k=1, k\neq j}^n (-1)^{\bigoplus_{i=k}^n t_j^i+a_j}\hat{\theta}_j^k
\end{equation} 
while for $j\in O^c\setminus I$, the computation oracle chooses random measurement angles $\hat{\theta}_j^k$ for clients $C_k,k\in[n-1]$ and sets:
\begin{equation}\label{app:angle2}
\hat{\theta}_j^n:=\delta_j-\phi'_j-\sum_{k=1}^{n-1} (-1)^{\bigoplus_{i=k}^{n-1} t_j^i}\hat{\theta}_j^k
\end{equation}
where undefined values are equal to zero, or otherwise:
\begin{itemize}
\item $\phi'_j=(-1)^{a_j+s_j^X}\phi_j+s^Z_j\pi+a_{f^{-1}(j)}\pi$.
\item $r_j=\bigoplus\limits_{k=1}^n r_j^k$.
\item $s_i=b_i\oplus r_i$, for $i\leq j$.
\end{itemize}
\item The clients measure the respective qubits in the received measurement bases.

\end{enumerate}
\item[output qubits:] For $j\in O$, the Server sends the ``encrypted'' quantum state to client $C_{j-q}$. All participants jointly compute $s_j^X$ and $s_j^Z$ and send it to client $C_{j-q}$, who applies operation $Z^{s_j^Z}X^{s_j^X}$ to retrieve the actual quantum output. 

\end{description}
\end{algorithm}


\clearpage 

\begin{algorithm}[!h]
\caption{Simulator for Server}
\label{app:sim}
\begin{description}
\item[non-output qubits:] For $j\in O^c$
\begin{enumerate}
\item $\sigma_S$ creates $n$ EPR pairs and sends one half of each to the Server.
\item $\sigma_S$ receives vector $\mathbf{t}_j$. 
\item $\sigma_S$ sends $\delta_j\in_R\{l\pi/4\}_{l=1}^7$ to the Server and receives a reply $b_j\in\{0,1\}$.
\end{enumerate}
\item[output qubits:] For $j\in O$
\begin{enumerate}
\item $\sigma_S$ receives $n$ qubits from the Server.
\item $\sigma_S$ sends the other halfs of the EPR pairs, the received quantum states, as well as $\delta_j$, $b_j$ and $\mathbf{t}_j$ for $j=1,\dots,q$, to the MPQC resource.
\end{enumerate}
\end{description}
\end{algorithm}

\begin{algorithm}[!h]
\caption{MPQC resource}
\label{app:mpqc}
\begin{enumerate}
\item The resource receives the $n$ qubits of $\rho_{in}$ from the clients, measurement angles $\{\phi_j\}_{j=1}^q$ and all the information from $\sigma_S$.
\item For $j\in I$: the resource performs a CNOT on the corresponding EPR half with the input qubit as control and measures the EPR half in the computational basis, getting result $a_j$. It chooses random measurement angles $\hat{\theta}_j^k$ for the qubits coming from clients $C_k,k\neq j$, sets:
\begin{equation*}
\hat{\theta}_j^j:=\delta_j-\phi'_j-\sum_{k=1, k\neq j}^n (-1)^{\bigoplus_{i=k}^n t_j^i+a_j}\hat{\theta}_j^k
\end{equation*} 
and measures the corresponding qubits. For $j\in O^c\setminus I$, it chooses random measurement angles $\hat{\theta}_j^k$ for clients $C_k,k\in[n-1]$ and sets:
\begin{equation*}
\hat{\theta}_j^n:=\delta_j-\phi'_j-\sum_{k=1}^{n-1} (-1)^{\bigoplus_{i=k}^{n-1} t_j^i}\hat{\theta}_j^k
\end{equation*}
where undefined values are equal to zero, or otherwise:
\begin{itemize}
\item $\phi'_j=(-1)^{a_j+s_j^X}\phi_j+s^Z_j\pi+a_{f^{-1}(j)}\pi$.
\item $r_j=\bigoplus\limits_{k=1}^n r_j^k$.
\item $s_i=b_i\oplus r_i$, for $i\leq j$.
\end{itemize}
\item For $j\in O$: the resource performs corrections $Z^{s_j^Z}X^{s_j^X}$ on the remaining qubits.
\end{enumerate}
\end{algorithm}

We can now check step-by-step the equivalence of the protocols described above and argue that  Eq.(\ref{secure1}) holds. We start by comparing Protocols \ref{app:protocol1} and \ref{app:protocol2}. In Protocol \ref{app:protocol1}, at Step \ref{1.1}, client $C_j$ chooses $a_j$ and $\theta_j^j$ uniformly at random from their domains and one-time pads his input state. In Protocol \ref{app:protocol2}, at Step \ref{2.1}, $C_j$ chooses uniformly at random $\hat{\theta}_j^j$ and teleports his input register to the Server, one-time padded with $X^{a_j}Z(\theta_j^j)$, where $\theta_j^j=\hat{\theta}_j^j+\pi r_j^j$. Since both $a_j$ and $\theta_j^j$ occur with the same probabilities, the state that the Server receives from client $C_j$ is the same in both protocols. Similarly in step \ref{1.2} of Protocol \ref{app:protocol1}  client  $C_k$ chooses uniformly at random $\theta_j^k$ and rotates the state $\ket{+}$ accordingly. In step \ref{2.2} of Protocol  \ref{app:protocol2} client $C_j$ chooses uniformly $\hat{\theta}_j^k$ and teleports to the Server the state $\ket{+}$ rotated by $\theta_j^k=\hat{\theta}_j^k+\pi r_j^k$. Since $\theta_j^k$ appears with the same probabilities for all clients in both protocols, the state described by Eq. \eqref{appeq:entangle1.1} and \eqref{appeq:entangle2.1} that the Server has, are the same. The same argumentation holds for step 4 of the two protocols, therefore at the end of the preparation phase, the Server has received exactly the same information from the clients. Finally, at the computation phase of the two protocols, the clients choose the measurement angles with the same probability.

We now check the equivalence of Protocols \ref{app:protocol2} and \ref{app:protocol3}. The main difference of Protocol \ref{app:protocol3} is that the phase flip (measurement of $r_j^k$) is delayed till after the measurement of the half of EPR pair by the Server. This is possible because the operation commutes with the teleportation. The states that the Server holds in both protocols are the same due to no signaling. In the computation phase, in Protocol \ref{app:protocol2}, the uniformly random value $\hat{\theta}_j^k$ defines the measurement angle $\delta_j$, while in \ref{app:protocol3} the uniformly random value $\delta_j$ defines $\hat{\theta}_j^k$ and thus the delayed step of the teleportation.

Finally, the combined simulator and Ideal Resource defined in Protocols \ref{app:sim} and \ref{app:mpqc} are just a separation and renaming of the preparation and computation tasks that the clients are required to do. It is easy to see that the Ideal Resource described in Protocol \ref{app:mpqc} fits the requirements of the MPQC resource defined in the main text, and therefore we have proven that the communication Protocol is equivalent to the Ideal Resource and a simulator for a dishonest Server: $\pi_1\dots\pi_n\mathcal{O}= \sigma_S \mathcal{M}$.

\vspace{0.5in}

\end{document}